%% file: main.tex
\documentclass[letterpaper,runningheads,envcountsame,10pt]{article}
\usepackage[left=1.2in,top=1.2in,right=1.2in,bottom=1.2in]{geometry}
\usepackage[utf8]{inputenc}
\usepackage{graphicx}
\usepackage{amsmath,amssymb}
\usepackage{amsthm}
\usepackage{ascmac}
\usepackage{cleveref}
\usepackage{authblk}

\usepackage{sectsty}
\sectionfont{\fontsize{12}{15}\selectfont}
\usepackage{algorithm}
\usepackage{algorithmicx}
\usepackage[noend]{algpseudocode}

\theoremstyle{definition}
\newtheorem{theorem}{Theorem}
\newtheorem{lemma}{Lemma}

\newtheorem{proposition}{Proposition}
\newtheorem{definition}{Definition}
\newtheorem{branch}{Branching rule}
\crefname{branch}{Branching rule}{Branching rules}
\newtheorem{reduce}{Reduction rule}
\crefname{reduce}{Reduction rule}{Reduction rules}
\newtheorem{assume}{Assumption}
\crefname{assume}{Assumption}{Assumptions}

\newcommand{\cc}{{\tt cc}}
\newcommand{\bridge}{{\tt b}}

\title{An Improved Deterministic Parameterized Algorithm for Cactus Vertex Deletion}
\author[1]{Yuuki Aoike}
\author[2]{Tatsuya Gima}
\author[2]{Tesshu Hanaka}
\author[1]{Masashi Kiyomi}
\author[3]{Yasuaki Kobayashi}
\author[4]{Yusuke Kobayashi}
\author[5]{Kazuhiro Kurita}
\author[2]{Yota Otachi}

\affil[1]{School of Data Science, Yokohama City University, Kanagawa, Japan}
\affil[2]{Graduate School of Informatics, Nagoya University}
\affil[3]{Graduate School of Informatics, Kyoto University, Kyoto, Japan}
\affil[4]{Research Institute for Mathematical Sciences, Kyoto University, Kyoto, Japan}
\affil[5]{National Institute of Informatics, Tokyo, Japan}
\date{}

\begin{document}

\maketitle

\begin{abstract}
    A {\em cactus} is a connected graph that does not contain $K_4 - e$ as a minor.
    Given a graph $G = (V, E)$ and an integer $k \ge 0$, {\sc Cactus Vertex Deletion} (also known as {\sc Diamond Hitting Set}) is the problem of deciding whether $G$ has a vertex set of size at most $k$ whose removal leaves a forest of cacti.
    The previously best deterministic parameterized algorithm for this problem was due to Bonnet et al. [WG 2016], which runs in time $26^kn^{O(1)}$, where $n$ is the number of vertices of $G$.
    In this paper, we design a deterministic algorithm for {\sc Cactus Vertex Deletion}, which runs in time $17.64^kn^{O(1)}$.
    As an almost straightforward application of our algorithm, we also give a deterministic $17.64^kn^{O(1)}$-time algorithm for {\sc Even Cycle Transversal}, which improves the previous running time $50^kn^{O(1)}$ of the known deterministic parameterized algorithm due to Misra et al. [WG 2012].
\end{abstract}
\input{body}

\end{document}

%% file: body.tex
\section{Introduction}

A connected graph is a {\em cactus} if every edge belongs to at most one cycle.
A {\em cactus forest} is a graph such that every connected component is a cactus.
In this paper, we consider the following problem.

\begin{definition}[{\sc Cactus Vertex Deletion}]
    Given a graph $G = (V, E)$ and an integer $k \ge 0$, the problem asks whether $G$ has a vertex set $X \subseteq V$ with $|X| \le k$ whose removal leaves a cactus forest.
\end{definition}

The problem is one of {\em vertex deletion problems for hereditary properties}, which have been both intensively and extensively studied in the field of parameterized algorithms and complexity.
The best known problem in this context is {\sc Vertex Cover}.
The problem asks whether there is a vertex set of size at most $k$ whose removal leaves an edge-less graph.
A naive algorithm solves {\sc Vertex Cover} in $O^*(2^k)$ time\footnote{The notation $O^*$ suppresses a polynomial factor of the input size.}, and after a series of improvements, the fastest known algorithm is due to Chen et al. \cite{ChenKX10:Improved}, which runs in time $O^*(1.2738^k)$.

Another example of this kind of problems is {\sc Feedback Vertex Set}.
The problem asks whether an input graph $G = (V, E)$ has a vertex set of size at most $k$ that hits all the cycles in the graph.
In other words, the goal of this problem is to compute $X \subseteq V$ with $|X| \le k$ such that the graph obtained from $G$ by deleting $X$ is a forest.
This problem is also intensively studied, and several deterministic and randomized algorithms have been proposed so far~\cite{Becker:Randomized:2000,Cygan:Solving:2011,Downey:Fixed:1995,Guo:Compression:2006,IwataK19:Improved,LiN20:Detecting}.
The current best running time is due to Iwata and Kobayashi~\cite{IwataK19:Improved} for deterministic algorithms and Li and Nederlof~\cite{LiN20:Detecting} for randomized algorithms, which run in time $O^*(3.460^k)$ and $O^*(2.7^k)$, respectively. 

The gap between the running time of deterministic and randomized algorithms sometimes emerges for vertex deletion problems to ``sparse'' hereditary classes of graphs, such as {\sc Feedback Vertex Set}.
For instance, {\sc Pseudo Forest Vertex Deletion} can be solved deterministically in time $O^*(3^k)$~\cite{BodlaenderOO18:faster} and randomizedly in time $O^*(2.85^k)$~\cite{GowdaLPPS20:Improved} and {\sc Bounded Degree-$2$ Vertex Deletion} can be solved deterministically in time $O^*(3.0645^k)$~\cite{Xiao16:parameterized} and randomizedly in time $O^*(3^k)$~\cite{FengWLC15:Randomized}.
Among others, the known gap on {\sc Cactus Vertex Deletion} is remarkable: Bonnet et al.~\cite{BonnetBKM16:Parameterized} presented a deterministic $O^*(26^k)$-time algorithm, while Kolay et al.~\cite{KolayLPS17:Quick} presented a randomized $O^*(12^k)$-time algorithm.

In this paper, we narrow the gap between the running time of deterministic and randomized algorithms by giving an improved deterministic algorithm for {\sc Cactus Vertex Deletion}.
\begin{theorem}\label{thm:CVD}
    {\sc Cactus Vertex Deletion} can be solved deterministically in time $O^*(17.64^k)$.
\end{theorem}

As a variant of {\sc Cactus Vertex Deletion}, we consider {\sc Even Cycle Transversal} defined as follows.
A cactus is called an {\em odd cactus} if every cycle in it has an odd number of vertices.

\begin{definition}[{\sc Even Cycle Transversal}]
    Given a graph $G = (V, E)$ and an integer $k \ge 0$, the problem asks whether $G$ has a vertex set $X \subseteq V$ with $|X| \le k$ whose removal leaves a forest of odd cacti.
\end{definition}

Note that a graph has no cycles of even length if and only if it is a forest of odd cacti~\cite{KolayLPS17:Quick}.
Kolay et al.~\cite{KolayLPS17:Quick} gave an $O^*(12^k)$-time randomized algorithm and Misra et al.~\cite{MisraRRS12:Parameterized} gave an $O^*(50^k)$-time deterministic algorithm for {\sc Even Cycle Transversal}.
In this paper, we improve the running time of the deterministic algorithm for {\sc Even Cycle Transversal} as well as {\sc Cactus Vertex Deletion}.
\begin{theorem}\label{thm:d-CVD-ECT}
    {\sc Even Cycle Transversal} can be solved deterministically in time $O^*(17.64^k)$.
\end{theorem}

The idea of our algorithms follows that used in \cite{BonnetBKM16:Parameterized}.
We solve the disjoint version of {\sc Cactus Vertex Deletion} with a branching algorithm.
In this version, given a vertex subset $S \subseteq V$ such that $|S| \le k + 1$ and the subgraph induced by $V \setminus S$, denoted $G[V \setminus S]$, is a cactus forest, the problem asks whether there is a vertex subset $X \subseteq V \setminus S$ such that $|X| \le k$ and $G[V \setminus X]$ is a cactus forest.
To solve this problem, Bonnet et al.~\cite{BonnetBKM16:Parameterized} gave a branching algorithm with the measure and conquer analysis~\cite{FominGK09:measure}.
They used measure $k + \cc(G[S])$, where $\cc(G[S])$ is the number of connected components in $G[S]$, and proved that each branch of their algorithm strictly decreases this measure.
The main difficulty with using this measure is that when we consider a vertex $v \in V \setminus S$ such that $v$ has at least two neighbors only in a single connected component in $G[S]$, then one of the branch, for which $v$ is determined to be not deleted, does not decrease the measure. 
We also use the measure and conquer analysis with a slightly elaborate measure $\alpha k + \beta \cdot \cc(G[S]) + \gamma\cdot\bridge(G[S])$, where $\alpha,\beta,\gamma$ are some constants and $\bridge(G[S])$ is the number of bridges in $G[S]$, which allows us to decrease the measure efficiently: When $v$ is determined to be not deleted, the number of bridges in $G[S]$ is decreased in the above situation since otherwise $v$ belongs to a $K_4 - e$ minor.
We believe that although our measure is slightly involved compared to that in \cite{BonnetBKM16:Parameterized}, the algorithm itself and its analysis would be simpler than theirs.

\section{Preliminaries}
\paragraph{Graphs.}
Throughout the paper, graphs have no self-loops but may have multiedges.
Let $G = (V, E)$ be a graph.
We write $V(G)$ and $E(G)$ to denote the sets of vertices and edges of $G$, respectively.
For two distinct vertices $u, v$ in $G$, we denote by $m(u, v)$ the number of edges between $u$ and $v$.
Let $v \in V$. 
The {\em degree} of $v$ is the number of edges incident to it.
We denote by $N_G(v)$ the set of neighbors of $v$ in $G$.
Note that as $G$ may have multiedges, $|N_G(v)|$ may not be equal to its degree.
For $X \subseteq V$, the subgraph of $G$ induced by $X$ is denoted as $G[X]$.
We denote by $\cc(G)$ the number of connected components in $G$.
A vertex $v \in V$ is called a {\em cut vertex} of $G$ if $\cc(G[V \setminus \{v\}]) > \cc(G)$ and an edge $e \in E$ is called a {\em bridge} of $G$ if $\cc(G - e) > \cc(G)$, where $G - e$ is the graph obtained from $G$ by deleting $e$.
Note that it holds that $\cc(G - e) = \cc(G) + 1$ when $e$ is a bridge of $G$.
The number of bridges in $G$ is denoted by $\bridge(H)$.
\begin{lemma}\label{lem:cc+b}
    Let $H$ be a multigraph with $h$ vertices.
    Then, it holds that $\cc(H) + \bridge(H) \le h$.
\end{lemma}
\begin{proof}
    Let $H'$ be the graph obtained from $H$ by removing all bridges of $H$.
    Then, $\cc(H) + \bridge(H) = \cc(H') \le h$.
\qed\end{proof}

A {\em block} of a graph $G$ is a maximal vertex set $B$ of $G$ such that $G[B]$ is connected and has no cut vertices.
Note that a graph consisting of two vertices with at least one edge is a block.
It is easy to see that every block in a cactus forest is either a cycle, an edge, or an isolated vertex.
In particular, we call $B$ a {\em leaf block} if it has at most one cut vertex.
We say that vertices $v_1, \ldots, v_t \in V(B)$ are {\em consecutive} in $B$ if for each $1 \le i < t$, $v_i$ is adjacent to $v_{i+1}$ in $B$.

\paragraph{Iterative compression.}
Our algorithm employs the well-known {\em iterative compression} technique invented by Reed, Smith, and Vetta~\cite{ReedSV04:Finding}.
They gave an algorithm for {\sc Odd Cycle Transversal} based on this technique.
The essential idea can be generalized as follows.
Let $\mathcal C$ be a hereditary class of graphs, that is, for $G \in \mathcal C$, every induced subgraph of $G$ also belongs to $\mathcal C$.
The technique is widely used for designing algorithms of vertex deletion problems to hereditary classes of graphs.
The crux of the technique can be described as the following lemma.

\begin{lemma}[\cite{ReedSV04:Finding}]\label{lem:compression}
    Let $\mathcal C$ be a hereditary class of graphs.
    Given a graph $G = (V, E)$ and an integer $k$, the problem of computing $X \subseteq V$ with $|X| \le k$ such that $G[V \setminus X] \in \mathcal C$ can be solved in time $O^*((c + 1)^k)$ if one can solve the following problem in time $O^*(c^k)$:
    Given a subset $S \subseteq V$ of cardinality at most $k + 1$ with $G[V \setminus S] \in \mathcal C$, the problem asks to find $X \subseteq V \setminus S$ with $|X| \le k$ such that $G[V \setminus X] \in \mathcal C$.
\end{lemma}

For {\sc Cactus Vertex Deletion}, the latter problem is defined as follows.
\begin{definition}[{\sc Disjoint Cactus Vertex Deletion}]
    Given a graph $G = (V, E)$, an integer $k \ge 0$, and $S \subseteq V$ such that $G[V \setminus S]$ is a cactus forest, the problem asks to find a vertex set $X \subseteq V \setminus S$ with $|X| \le k$ whose removal leaves a cactus forest. 
\end{definition}

Let us note that we can assume that $G[S]$ is also a cactus forest as otherwise the problem is trivially infeasible.

\paragraph{Measure and conquer analysis.}
Our algorithm for {\sc Disjoint Cactus Vertex Deletion} is based on a standard branching algorithm with the measure and conquer analysis~\cite{FominGK09:measure}.
Given an instance $I$ of the problem, we define a measure $\mu(I)$ that is non-negative real and design a branching algorithm that generates subinstances $I_1, \ldots, I_t$ with $\mu(I) > \mu(I_i)$ for $1 \le i \le t$.
To measure the running time of the algorithm, we use a branching factor $(b_1, \ldots, b_t)$, where $\mu(I) - \mu(I_i) \ge b_i$ for each $i$.
It is known that the total running time of this branching algorithm is upper bounded by $O^*(c^{\mu(I)})$, where $c$ is the unique positive real root of equation
\begin{align*}
    x^{-b_1} + x^{-b_2} + \cdots + x^{-b_t} = 1,
\end{align*}
assuming that from any instance $I$ with $\mu(I) > 0$, its subinstances can be generated in polynomial time and for any instance $I$ with $\mu(I) = 0$, the problem can be solved in polynomial time.
We refer the reader to the book~\cite{FominD10:Exact} for a detailed exposition for the measure and conquer analysis.

\section{An improved algorithm for {\sc Disjoint Cactus Vertex Deletion}}

This section is devoted to developing an algorithm for {\sc Disjoint Cactus Vertex Deletion} that runs in time $O^*(16.64^k)$, proving Theorem~\ref{thm:CVD} by Lemma~\ref{lem:compression}. 
\begin{lemma}\label{lem:DCVD}
    Suppose that $|S| \le k + 1$.
    Then, {\sc Disjoint Cactus Vertex Deletion} can be solved in time $O^*(16.64^k)$.
\end{lemma}

Let $I = (G, S, k)$ be an instance of {\sc Disjoint Cactus Vertex Deletion}, where $G = (V, E)$ is a multigraph, $S \subseteq V$.
Recall that we assume $G[V \setminus S]$ and $G[S]$ are both cactus forests as otherwise the problem is trivially infeasible.

Let $\mu(I) = \alpha\cdot k + \beta\cdot\cc(G[S]) + \gamma\cdot\bridge(G[S])$, where $\alpha,\beta,\gamma$ are chosen later.
In the following, we assume that $\beta \ge \gamma$.
For the sake of simplicity, we write, for $X \subseteq V$, $\cc(X)$ and $\bridge(X)$ to denote $\cc(G[X])$ and $\bridge(G[X])$, respectively.

As $G$ may have multiedges, every cactus forest can be characterized as the following form.
\begin{proposition}[\cite{FioriniJP10:Hitting}]
    Let $D$ be the graph of two vertices and three parallel edges between them.
    A graph is a cactus forest if and only if it does not contain a subgraph isomorphic to any subdivision of $D$.
\end{proposition}

We call a subdivision of $D$ an {\em obstruction}.
In particular, $D$ itself is also an obstruction.

The algorithm consists of several branching rules and reduction rules.
We say that a reduction rule is {\em safe} if the original instance has a yes-instance if and only if the instance obtained by applying the rule is a yes-instance.
We also say that a branching rule is {\em safe} if the original instance is a yes-instance if and only if at least one of the instances obtained by applying the rule is a yes-instance.
Our algorithm described below {\em determines} whether $G$ has a solution $X$ for {\sc Disjoint Cactus Vertex Deletion}.
However, the algorithm easily turns into one that {\em finds} an actual solution if the answer is affirmative.
We apply these rules in the order of their appearance.
The algorithm terminates if $V(G) = S$ or $k = 0$, and it answers ``YES'' if and only if $k \ge 0$ and $G$ is a cactus forest.

The following reduction and branching rules are trivially safe.
\begin{reduce}\label{rr:isolated-component}
    If $G[V \setminus S]$ contains a component $C$ that has no neighbors in $S$, then delete all the vertices in $C$.
\end{reduce}
\begin{reduce}\label{rr:deg-1}
    If $G[V \setminus S]$ contains a vertex of degree one in $G$, then delete it.
\end{reduce}
\begin{reduce}\label{rr:must-be-in-solution}
    If $G[V \setminus S]$ contains a vertex $v$ such that $G[S \cup \{v\}]$ is not a cactus forest, then delete $v$ and decrease $k$ by one.
\end{reduce}
\begin{branch}\label{br:must-be-in-solution}
    If $G[V \setminus S]$ contains vertices $u, v \in V \setminus S$ with $m(u, v) \ge 3$, branch into two cases: (1) delete $u$ and decrease $k$ by one; (2) delete $v$ and decrease $k$ by one.
\end{branch}

The branching factor of Branching rule~\ref{br:must-be-in-solution} is $(\alpha, \alpha)$.
By applying these rules, we make the following assumption on each vertex in $V \setminus S$.
\begin{assume}
    Every vertex $v \in V \setminus S$ has degree at least two in $G$ and there are at most two edges between two vertices.
\end{assume}

As $G$ is a multigraph, some vertex may have only one neighbor even if its degree is greater than one.
If $G[V \setminus S]$ contains a vertex $v$ with $|N_G(v)| = 1$, this vertex also can be removed since it is not a part of an obstruction, assuming that $m(u, v) \le 2$ with $u \in N_G(v)$.
This implies the following reduction rule.

\begin{reduce}\label{rr:neighbor-1}
    If $G[V \setminus S]$ contains a vertex $v$ with $|N_G(v)| = 1$, then delete it.
\end{reduce}
Thus, we further make the following assumption on each vertex in $V \setminus S$.

\begin{assume}\label{assume:al2}
    Every vertex $v \in V \setminus S$ has at least two neighbors in $G$.
\end{assume}

Suppose that there is a vertex $v \in V \setminus S$ that has at least two neighbors in $S$.
By Reduction rule~\ref{rr:must-be-in-solution}, there is no component in $G[S]$ that contains at least three vertices of $N_G(v) \cap S$.
Let $W = N_G(v) \cap S$.
We denote by $t_1$ (resp. by $t_2$) the number of components in $G[S]$ that contain exactly one vertex (resp. two vertices) of $W$.
Let $C$ be a component in $G[S]$ that has at least one vertex of $W$.
If $|W \cap C| = 2$, say $w, w' \in W \cap C$, every edge on the path between $w$ and $w'$ in $G[C]$ is a bridge as otherwise $G[C \cup \{v\}]$ contains an obstruction, which implies that $v$ is removed by Reduction rule~\ref{rr:must-be-in-solution}. 
Then, there is at least one bridge on the path between $w$ and $w'$ in $G[C]$.
Thus, $\bridge(C \cup \{v\}) \le \bridge(C) - 1$.
If $|W \cap C| = 1$, $G[C \cup \{v\}]$ has $\bridge(C) + 1$ bridges.
Hence, we have
\begin{align*}
    \beta \cdot \cc(S \cup \{v\}) &\le \beta \cdot \cc(S) - \beta(t_1 + t_2 - 1), \\
    \gamma \cdot \bridge(S \cup \{v\}) & \le \gamma \cdot\bridge(S) + \gamma(t_1 - t_2).
\end{align*}

Consider the value $t_1(\beta - \gamma) + t_2(\beta + \gamma) - \beta$, that is, a lower bound of $\mu((G, S, k)) - \mu((G, S \cup \{v\}, k))$.
If $t_1 + t_2 \ge 2$, the value is at least $\beta - 2\gamma$.
This follows from the fact that the value is minimized when $t_1 = 2$ and $t_2 = 0$ under $\beta \ge \gamma \ge 0$.
If $t_1 + t_2 = 1$, $t_1$ must be zero since $|W| \ge 2$.
In this case, the value is at least $\gamma$.
This implies the following branching rule, which is clearly safe, has branching factor $(\alpha, \min(\beta - 2\gamma, \gamma))$.

\begin{branch}\label{br:v-has-two-neighbors-in-two-components}
    Suppose $G[V \setminus S]$ contains a vertex $v$ that has at least two neighbors in $S$.
    Then, branch into two cases: (1) delete $v$ and decrease $k$ by one; (2) put $v$ into $S$.
\end{branch}

Thus, we make the following assumption on each vertex in $V \setminus S$.

\begin{assume}\label{assume:at-most-1-in-S}
    Every vertex $v \in V \setminus S$ has at least two neighbors in $G$ and at most one of them belongs to $S$.
\end{assume}

We can remove a vertex having exactly two neighbors by adding an edge between its neighbors.
The following lemma justifies this reduction.

\begin{lemma}\label{lem:eliminate-deg2}
    Let $v \in V \setminus S$ be a vertex with exactly two neighbors $u, w$ in $G$.
    Suppose that $p = \max(m(u, v), m(v, w)) \le 2$.
    Let $G'$ be the graph obtained from $G$ by deleting $v$ and adding $p$ parallel edges between $u$ and $w$.
    Then, $G$ has a cactus deletion set of size at most $k$ if and only if $G'$ has a cactus deletion set of size at most $k$.
\end{lemma}

\begin{proof}
    Since every obstruction in $G$ containing $v$ also has both $u$ and $w$, there is a smallest cactus deletion set $X$ that does not contain $v$.
    Such a set is also a cactus deletion set of $G'$ and vise versa.
\end{proof}

By Lemma~\ref{lem:eliminate-deg2}, the following reduction rule is safe.

\begin{reduce}\label{rr:deg-2}
    Suppose that $G[V \setminus S]$ contains a vertex $v$ with $N_G(v) = \{u, w\}$.
    Then delete $v$ and add $\max(m(u, v), m(v, w))$ parallel edges between $u$ and $w$.
\end{reduce}

This implies that the following assumption is made.

\begin{assume}\label{assume:deg-3}
    Every vertex $v \in V \setminus S$ has at least three neighbors.
    Moreover, at most one of them belongs to $S$ and hence at least two of them belong to $V \setminus S$.
\end{assume}

Since $G[V \setminus S]$ is a cactus forest, there is a leaf block $B$.
By Assumption~\ref{assume:deg-3}, $B$ contains at least three vertices.
Suppose that $B$ has exactly three vertices $u, v, w$.
As $B$ is a leaf block, we can assume that both $u$ and $v$ are not cut vertices of $G[V \setminus S]$.
By Assumption~\ref{assume:at-most-1-in-S}, both $u$ and $v$ have exactly one neighbor in $S$, which can be an identical vertex.
If there is a component $C$ in $G[S]$ that contains both a neighbor of $u$ and a neighbor of $v$, then $G[C \cup \{u, v, w\}]$ has an obstruction, which yields the following branching rule with branching factor $(\alpha, \alpha, \alpha)$.

\begin{branch}\label{br:k3-single-comp}
    Suppose that there is a leaf block $B$ with $V(B) = \{u, v, w\}$ in $G[V \setminus S]$.
    Suppose moreover that each of $u$ and $v$ has exactly one neighbor in $S$ and that these neighbors belong to a single component in $G[S]$.
    Then, branch into three cases: (1) delete $u$; (2) delete $v$; (3) delete $w$. 
    For each case, decrease $k$ by one.
\end{branch}

Otherwise, the neighbors of $u$ and $v$ belong to distinct components in $G[S]$.
Let $C_u$ and $C_v$ be the components of $G[S]$ that have neighbors of $u$ and $v$, respectively. 
If $w$ has a neighbor in $C_u$ or $C_v$, then $G[S \cup \{u, v, w\}]$ contains an obstruction.
In this case, we apply Branching rule~\ref{br:k3-single-comp} as well.
Thus, either $w$ has no neighbor in $S$ or $w$ has exactly one neighbor in a component $C_w$ in $G[S]$ with $C_w \neq C_u$ and $C_w \neq C_v$.
In both cases, we apply the following branching rule.

\begin{branch}\label{br:k3-distinct-comp}
    Suppose that there is a leaf block $B$ with $V(B) = \{u, v, w\}$ in $G[V \setminus S]$.
    Suppose moreover that each of $u$ and $v$ has exactly one neighbor in $S$ and that these neighbors belong to distinct components in $G[S]$.
    Then, branch into four cases: (1) delete $u$; (2) delete $v$; (3) delete $w$; (4) put $u$, $v$, and $w$ into $S$. 
    For (1), (2), and (3), decrease $k$ by one.
\end{branch}

To see the branching factor of this rule, suppose first that $w$ has no neighbor in $S$.
Then, $G[S \cup \{u, v, w\}]$ contains $\cc(S) - 1$ components and $\bridge(S) + 2$ bridges.
Thus, Branching~rule~\ref{br:k3-distinct-comp} has branching factor $(\alpha, \alpha, \alpha, \beta - 2\gamma)$.
Suppose otherwise that $w$ has an exactly one neighbor in a component $C_w$ in $G[S]$ with $C_w \neq C_u$ and $C_w \neq C_v$.
Then, $G[S \cup \{u, v, w\}]$ contains $\cc(S) - 2$ components and $\bridge(S) + 3$ bridges.
Thus, Branching~rule~\ref{br:k3-distinct-comp} has branching factor $(\alpha, \alpha, \alpha, 2\beta - 3\gamma)$.

By Branching rules~\ref{br:k3-single-comp} and \ref{br:k3-distinct-comp}, the following assumption is made.

\begin{assume}\label{assume:leaf-block-size-4}
    Every leaf block $B$ in $G[V \setminus S]$ contains three consecutive vertices, each of which is not a cut vertex in $G[V \setminus S]$ and has exactly one neighbor in $S$.
\end{assume}

Let $u, v, w$ be three consecutive vertices in $B$, each of which is not a cut vertex in $G[V \setminus S]$ and has exactly one neighbor in $S$.
Let $u', v', w'$ be the neighbors of $u, v, w$ in $S$, respectively.
There are four cases~(Figure~\ref{fig:c4}).

\begin{figure*}
    \centering
    \includegraphics[width=\textwidth]{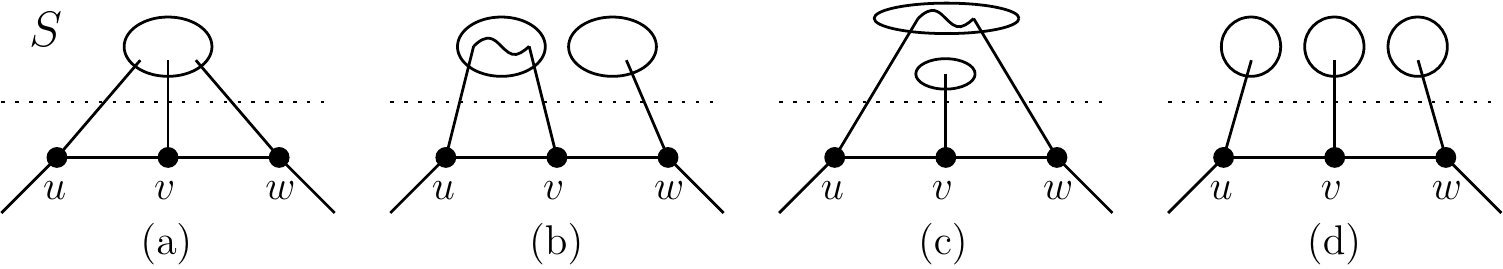}
    \caption{An Illustration of four branching cases under Assumption~\ref{assume:leaf-block-size-4}.}
    \label{fig:c4}
\end{figure*}

Suppose that there is a component $C$ in $G[S]$ that contains these neighbors ((a) in Figure~\ref{fig:c4}).
Then, $G[C \cup \{u, v, w\}]$ has an obstruction, yielding the following Branching rule~\ref{br:non-k3-single-comp} that has branching factor $(\alpha, \alpha, \alpha)$.

\begin{branch}\label{br:non-k3-single-comp}
    Suppose that there is a leaf block $B$ with $|V(B)| \ge 4$ in $G[V \setminus S]$.
    Let $u, v, w$ be three consecutive vertices in $B$, each of which is not a cut vertex in $G[V \setminus S]$ and has exactly one neighbor in $S$.
    Suppose that these neighbors belong to a single component in $G[S]$.
    Then, branch into three cases: (1) delete $u$; (2) delete $v$; (3) delete $w$. 
    For each case, decrease $k$ by one.
\end{branch}

Suppose next that exactly two of $u', v', w'$ are contained in a single component $C$ in $G[S]$.
There are essentially two cases: (1) $u' $ and $v'$ are contained in $C$ ((b) in Figure~\ref{fig:c4}) or (2) $u'$ and $w'$ are contained in $C$ ((c) in Figure~\ref{fig:c4}).
In case (1), $G[S \cup \{u, v, w\}]$ contains $\cc(S) - 1$ components and $\bridge(S) + 2$ bridges.
In case (2), $G[S \cup \{u, v, w\}]$ contains $\cc(S) - 1$ components and $\bridge(S) + 1$ bridges.
For these cases, we apply the following Branching~rule~\ref{br:non-k3-distinct-comp}, which has branching factors $(\alpha, \alpha, \alpha, \beta - 2\gamma)$ and $(\alpha, \alpha, \alpha, \beta - \gamma)$ for these cases.

\begin{branch}\label{br:non-k3-distinct-comp}
    Suppose that there is a leaf block $B$ with $|V(B)| \ge 4$ in $G[V \setminus S]$.
    Let $u, v, w$ be three consecutive vertices in $B$, each of which is not a cut vertex in $G[V \setminus S]$ and has exactly one neighbor in $S$.
    Suppose that these neighbors are not contained in a single component in $G[S]$.
    Then, branch into four cases: (1) delete $u$; (2) delete $v$; (3) delete $w$; (4) put $u$, $v$, and $w$ into $S$. 
    For (1), (2), and (3), decrease $k$ by one.
\end{branch}

Finally, suppose any two of $u', v', w'$ are not contained in a single component in $G[S]$ ((d) in Figure~\ref{fig:c4}).
Again, we apply Branching~rule~\ref{br:non-k3-distinct-comp} to this case.
Since $G[S \cup \{u, v, w\}]$ contains $\cc(S) - 2$ components and $\bridge(S) + 5$ bridges, Branching~rule~\ref{br:non-k3-distinct-comp} has branching factor $(\alpha, \alpha, \alpha, 2\beta - 5\gamma)$.
The entire algorithm for {\sc Disjoint Cactus Vertex Deletion} is given in Algorithm~\ref{alg:code}.

The reduction and branching rules cover all cases for the instance $I$ and all the rules are safe.
Thus, the algorithm correctly computes a cactus deletion set $X \subseteq V \setminus S$ with $|X| \le k$ if it exists.
By choosing $\alpha = 1$, $\beta = 0.4052$, $\gamma = 0.0726$, the running time is dominated by the branching factor $(\alpha, \alpha, \alpha, \beta - 2\gamma) = (1, 1, 1, 0.26)$.
By Lemma~\ref{lem:cc+b}, we have $\beta \cdot \cc(S) + \gamma \cdot \bridge(S) \le \beta (k + 1)$. 
Therefore, the running time of the algorithm is
\begin{align*}
    O^*(c^{\mu(I)}) &\subseteq O^*(c^{\alpha \cdot k + \beta \cdot \cc(S) + \gamma \cdot \bridge(S)}) \\
    &\subseteq O^*(c^{1.4052k}), 
\end{align*}
where $c < 7.3961$ is the unique positive real root of equation $3x^{-1} + x^{-0.26} = 1$.
This yields the running time bound $O^*(16.64^k)$ for {\sc Disjoint Cactus Vertex Deletion}.

\begin{algorithm}[h]
\caption{A pseudocode of the algorithm for {\sc Disjoint Cactus Vertex Deletion}}\label{alg:code}
\begin{algorithmic}
    \Procedure{\tt DCVD}{$G = (V, E), S, k$}
    \If{$k \ge 0$ and $V = S$}
        \State {\bf return} {\bf true}
    \EndIf
    \If{$k < 0$}
        \State {\bf return} {\bf false}
    \EndIf
    \If{$G[V \setminus S]$ has a component $C$ that has no neighbors in $S$} \Comment{Reduction~rule~\ref{rr:isolated-component}}
        \State {\bf return} {\tt DCVD}$(G[V \setminus C], S, k)$ 
    \EndIf
    \If{$G[V \setminus S]$ has a vertex $v$ of degree one in $G$} \Comment{Reduction~rule~\ref{rr:deg-1}}
        \State {\bf return} {\tt DCVD}$(G[V \setminus \{v\}], S, k)$
    \EndIf
    \If{$G[V \setminus S]$ has $v$ such that $G[V \cup \{v\}]$ is not a cactus forest} \Comment{Reduction~rule~\ref{rr:must-be-in-solution}}
        \State {\bf return} {\tt DCVD}$(G[V \setminus \{v\}], S, k - 1)$
    \EndIf
    \If{$G[V \setminus S]$ has vertices $u$ and $v$ with $m(u, v) \ge 3$} \Comment{Branching~rule~\ref{br:must-be-in-solution}}
        \State {\bf return} {\tt DCVD}$(G[V \setminus \{u\}], S, k - 1) \lor $ {\tt DCVD}$(G[V \setminus \{v\}], S, k - 1)$ 
    \EndIf
    \If{$G[V \setminus S]$ has a vertex $v$ with $|N_G(v)| = 1$} \Comment{Reduction~rule~\ref{rr:neighbor-1}}
        \State {\bf return} {\tt DCVD}$(G[V \setminus \{v\}], S, k)$ 
    \EndIf
    \If{$G[V \setminus S]$ has a vertex $v$ having at least two neighbors in $S$} \Comment{Branching~rule~\ref{br:v-has-two-neighbors-in-two-components}}
        \State {\bf return} {\tt DCVD}$(G[V \setminus \{v\}], S, k - 1) \lor $ {\tt DCVD}$(G[V], S \cup \{v\}, k)$ 
    \EndIf
    \If{$G[V \setminus S]$ has a vertex $v$ with $N_G(v) = \{u, w\}$} \Comment{Reduction~rule~\ref{rr:deg-2}}
        \State Let $G' = G[V \setminus \{v\}]$.
        \State Add $\max\{m(u, v), m(v, w)\}$ parallel edges between $u$ and $w$ to $G'$.
        \State {\bf return} {\tt DCVD}$(G', S, k)$
    \EndIf
     \If{$G[V \setminus S]$ has a leaf block $B$ with $V(B) = \{u, v, w\}$} \Comment{Branching~rules~\ref{br:k3-single-comp} and \ref{br:k3-distinct-comp}}
        \For{$x \in V(B)$}
            \If{{\tt DCVD}$(G[V \setminus \{x\}], S, k - 1)$}
                \State{\bf return true} 
            \EndIf
        \EndFor
        \If{$G[S \cup V(B)]$ is a cactus forest}
            \State {\bf return} {\tt DCVD}$(G[V], S \cup V(B), k)$
        \EndIf
        \State {\bf return false}
    \EndIf
    \If{$G[V \setminus S]$ has a leaf block $B$ with $|V(B)| \ge 4$} \Comment{Branching rules \ref{br:non-k3-single-comp} and \ref{br:non-k3-distinct-comp}}
        \State Let $B' = \{u, v, w\}$ be consecutive vertices in $B$ that are not cut vertices in $G[V \setminus S]$.
        \For{$x \in V(B')$}
            \If{{\tt DCVD}$(G[V \setminus \{x\}], S, k - 1)$}
                \State {\bf return true} 
            \EndIf
        \EndFor
        \If{$G[S \cup V(B')]$ is a cactus forest}
            \State {\bf return} {\tt DCVD}$(G[V], S \cup V(B'), k)$
        \EndIf
        \State {\bf return false}
    \EndIf
    \EndProcedure
\end{algorithmic}
\end{algorithm}

\section{An improved algorithm for {\sc Even Cycle Transversal}}

Recall that {\sc Even Cycle Transversal} asks whether, given a graph $G = (V, E)$ and an integer $k$, $G$ has a vertex set $X$ of size at most $k$ such that $G[V \setminus X]$ is a forest of odd cacti. 
As in the previous section, we solve the disjoint version of {\sc Even Cycle Transversal} and give an $O(16.64^k)$-time algorithm for it, assuming that $|S| \le k + 1$.
\begin{definition}[{\sc Disjoint Even Cycle Transversal}]
    Given a graph $G = (V, E)$, an integer $k \ge 0$, and $S \subseteq V$ such that $G[V \setminus S]$ is a forest of odd cacti, the problem asks to find a vertex set $X \subseteq V \setminus S$ with $|X| \le k$ whose removal leaves a forest of odd cacti.
\end{definition}

A key difference from {\sc Disjoint Cactus Vertex Deletion} is that we need to take the length of cycles into account.
However, in Reduction~rule~\ref{rr:deg-2}, we replace (a chain of) cycles with two multiple edges between two extreme vertices, which does not preserve the length of cycles in the original graph.
Given this, we consider a slightly generalized problem.
In addition to the input of {\sc Disjoint Even Cycle Transversal}, we are given a binary weight function $\omega\colon E \to \mathbb \{0, 1\}$ on edges, and the length of a cycle is defined to be the total weight of edges in it.
Indeed, when $\omega(e) = 1$ for all $e \in E$, the problem corresponds to {\sc Disjoint Even Cycle Transversal}.

Let $S \subseteq V$ such that $G[V \setminus S]$ is a forest of odd cacti.
We first apply Reduction~rules~\ref{rr:isolated-component}, \ref{rr:deg-1}, \ref{rr:must-be-in-solution}, and \ref{rr:neighbor-1} and Branching~rules~\ref{br:must-be-in-solution} and \ref{br:v-has-two-neighbors-in-two-components}, which are trivially safe for {\sc Disjoint Even Cycle Transversal}.
Moreover, we add the following reduction rule, which is also trivially safe.

\begin{reduce}\label{rr:ECT:must-be-in-solution}
    If there is a vertex $v \in V \setminus S$ such that $G[S \cup \{v\}]$ has a cycle of even length, then delete it and decrease $k$ by one.
\end{reduce}

We can check in linear time whether an edge-weighted graph is a forest of odd cacti and hence Reduction~rule~\ref{rr:ECT:must-be-in-solution} can be applied in linear time as well.

% It is not obvious to find a vertex in $V \setminus S$ that satisfies the condition in Reduction~rule~\ref{rr:ECT:must-be-in-solution}.
% An important observation to do this is that $G[S]$ has treewidth at most two.
% We refer to \cite{Cygan:Parameterized:2015} for the detailed definition of treewidth and its algorithmic usage.
% More specifically, we have the following proposition, which allows us to check the condition in linear time.

% \begin{proposition}\label{prop:finding-invalid-cycles}
%     Let $S \subseteq V$ that induces a cactus forest and let $X \subseteq V \setminus S$.
%     Suppose that $X$ has a constant number of vertices.
%     Then, we can check in linear time whether $G[S \cup X]$ has a cycle of even length.
% \end{proposition}

Up to this point, Assumption~\ref{assume:at-most-1-in-S} is made.
By Reduction~rule~\ref{rr:must-be-in-solution} and Branching~rule~\ref{br:must-be-in-solution}, we also assume that $m(u, v) \le 2$ for every pair of vertices in $G$.
Suppose that $m(u, v) = 2$.
If the parities of two edges between $u$ and $v$ are the same, the length of the cycle consisting of these edges is even.
Thus, we apply the following branching rule in this case.
\begin{branch}\label{br:even-K2}
    Suppose that $u, v \in V \setminus S$ and $m(u, v) = 2$ for some $u$.
    Let $f, f'$ be the edges between $u$ and $v$.
    If $\omega(f) = \omega(f')$, branch into two cases: (1) delete $u$ and decrease $k$ by one; (2) delete $v$ and decrease $k$ by one. 
\end{branch}

By Reduction~rule~\ref{rr:ECT:must-be-in-solution} and Branching~rule~\ref{br:even-K2}, the following assumption is made.
\begin{assume}\label{assume:parity}
    For every pair of vertices $u, v$ with $m(u, v) = 2$, the parities of the weights of edges between them are opposite.
\end{assume}

Now, let us consider a vertex $v \in V \setminus S$ that has exactly two neighbors in $G$.
Let $u$ and $w$ be the neighbors of $v$.
By Assumption~\ref{assume:at-most-1-in-S}, at least one of $u$ and $w$ belongs to $V \setminus S$.
Similarly to Lemma~\ref{lem:eliminate-deg2}, we define a graph $G'$ by deleting $v$ from $G$ and adding $p$ parallel edges between $u$ and $v$, where $p = \max(m(u, v), m(v, w))$.
We define the weight function $w'$ for $G'$ as follows.
If $p = 1$, we set the weight of the introduced edge $e = \{u, w\}$ as $\omega'(e) = \omega(f) + \omega(f')$, where $f$ (resp.~$f'$) is the edge between $u$ and $v$ (resp.~between $v$ and $w$) and the sum is taken under addition modulo two.
If $p = 2$, at least one of the pairs $\{u, v\}$ or $\{v, w\}$ has multiple edges.
By Assumption~\ref{assume:parity}, these two edges have different parities.
A crucial observation is that if there is a cycle passing through exactly one of these edges, there is another cycle passing through the other edges, which has the different parity. 
By setting $\omega'(e) = 0$ and $\omega'(e') = 1$, such cycles are preserved in $G'$. 

\begin{lemma}\label{lem:ECT:feasible}
    The instance $(G, \omega, S, k)$ is a yes-instance if and only if $(G', \omega', S, k)$ is a yes-instance.
\end{lemma}

\begin{proof}
    Consider a cycle $C$ of even length that passes through $v$ in $G$.
    By Assumption~\ref{assume:parity}, $C$ must pass through both $u$ and $w$.
    Thus, there is a feasible solution $X \subseteq V \setminus S$ for $(G, \omega, S, k)$ with $v \notin X$ and $\{u,w\} \cap X \neq \emptyset$.
    By the construction of $G'$, the cycle obtained from $C$ by omitting $v$ is an even cycle of $G'$.
    Hence, $X$ is a feasible solution for $(G', \omega', S, k)$.
    It is not hard to see that this correspondence is reversible and hence the lemma follows.
\end{proof}

This lemma ensures that the weighted version of Reduction~rule~\ref{rr:deg-2} is safe for {\sc Even Cycle Transversal}, and then Assumption~\ref{assume:deg-3} is made as well.
The rest of branching rules are the same with {\sc Disjoint Cactus Vertex Deletion}, which yields an $O^*(16.64^k)$-time algorithm that solves {\sc Disjoint Even Cycle Transversal} as well.

\begin{lemma}\label{lem:ECT:disjoint}
    Suppose $|S| \le k + 1$. Then, {\sc Disjoint Even Cycle Transversal} can be solved in time $O^*(16.64^k)$.
\end{lemma}

By Lemma~\ref{lem:compression}, {\sc Even Cycle Transversal} can be solved in time $O^*(17.64^k)$, completing the proof of Theorem~\ref{thm:d-CVD-ECT}.

\section*{Acknowledgments}
This work is partially supported by JSPS KAKENHI Grant Numbers JP18H04091, JP18H05291, JP18K11168, JP18K11169, JP19K21537, JP20H05793, JP20K11692, and JP20K19742.
The authors thank Kunihiro Wasa for fruitful discussions.

\bibliographystyle{plain}
\bibliography{main}